\documentclass[aip,reprint]{revtex4-1} %replace reprint with graphicx

\usepackage{graphicx}% Include figure files
\usepackage{dcolumn}% Align table columns on decimal point
\usepackage{bm}% bold math
\usepackage[utf8]{inputenc}
\usepackage[T1]{fontenc}
\usepackage{mathptmx}

\usepackage{amsmath,amsthm,amssymb}
\usepackage{graphicx}

\newcommand{\R}{{\mathbb R}}
\newcommand{\eps}{\varepsilon}

\usepackage{color}

\theoremstyle{plain}% Theorem-like structures
\newtheorem{theorem}{Theorem}[section]
\newtheorem{lemma}[theorem]{Lemma}

\newtheorem{proposition}[theorem]{Proposition}

\theoremstyle{definition}

\theoremstyle{remark}
\newtheorem{remark}{Remark}

\theoremstyle{assumption}
\newtheorem{assumption}{\bf Assumption}

\begin{document}

\title{One-Dimensional Lieb--Oxford Bounds} 

\author{Andre Laestadius}
\email[]{andre.laestadius@kjemi.uio.no}
%\homepage[]{http://folk.uio.no/andrelae/}
\affiliation{Hylleraas Centre for Quantum Molecular Sciences, Department of Chemistry, University of Oslo, P.O. Box 1033 Blindern, N-0315 Oslo, Norway}

\author{Fabian M. Faulstich}
%\email[]{f.m.faulstich@kjemi.uio.no}
%\homepage[]{http://folk.uio.no/fabianfa}
\affiliation{Hylleraas Centre for Quantum Molecular Sciences, Department of Chemistry, University of Oslo, P.O. Box 1033 Blindern, N-0315 Oslo, Norway}

\date{\today}

\begin{abstract}

We investigate and prove Lieb--Oxford bounds in one dimension by studying convex potentials that approximate the ill-defined Coulomb potential. 
A Lieb--Oxford inequality establishes a bound of the indirect interaction energy for electrons 
in terms of the one-body particle density $\rho_\psi$ of a wave function~$\psi$. Our results include modified soft Coulomb potential and regularized Coulomb potential. For these potentials,
we establish Lieb--Oxford-type bounds utilizing logarithmic expressions of the particle density. 
Furthermore, a previous conjectured form $I_\mathrm{xc}(\psi)\geq - C_1 \int_{\mathbb R} \rho_\psi(x)^{2} \mathrm{d}x$ is discussed for different convex potentials.

\end{abstract}

\maketitle

\section{Introduction}

Kohn--Sham density-functional theory \cite{kohn1965self} is due to its simplicity and wide ranging applicability today's workhorse of quantum many-body calculations. 
The general formulation of density-functional theory uses the one-body particle density~$\rho$, which can be computed from an antisymmetric wave function $\psi$ describing a fermionic many-body system.
Let $X$ be the space of admissible densities and let $\langle \cdot,\cdot \rangle_{X^*,X}$ be the dual pairing. 
Then the energy corresponding to $\rho$ can be expressed by
\begin{equation} \label{eq:dft}
\begin{aligned}
 E(\rho)
&=
\inf_{\psi \mapsto \rho} 
\langle \psi | \hat T +  \hat V | \psi \rangle + \langle v_\mathrm{ext},\rho \rangle_{X^*,X} \\
&= F(\rho) + \langle v_\mathrm{ext},\rho \rangle_{X^*,X},
\end{aligned}
\end{equation}
where $v_\mathrm{ext}$ is the external potential (element of the topological dual space $X^*$) and $\psi\mapsto \rho$ is the  conventional notation for $\rho_\psi = \rho$ with $\rho_\psi$ being the particle density computed from $\psi$. 
Here $\hat T$ and $\hat V$ describe the kinetic energy and electron--electron repulsion, respectively.

Equation~\eqref{eq:dft} describes the transition from a formulation using wave functions to a formulation using densities and where $F(\rho)$ denotes the universal density functional. 
For almost all practical applications, one more important step is taken:  
The Kohn--Sham approach introduces a fictitious non-interacting system that has the same ground-state particle density $\rho$ as the fully-interacting system and that can be computed from a Slater determinant.
We write
\begin{align*}
    F(\rho) = T_s(\rho) + J(\rho) +E_{\rm xc}(\rho), 
\end{align*}
where $T_\mathrm{s}(\rho)$ is the noninteracting kinetic energy, $J(\rho)$ is the direct Coulomb repulsion ({\it Hartree term)}, and $E_{\rm xc}(\rho)$ is the exchange-correlation energy. 
The electronic correlation effects are incorporated into the exchange-correlation functional.
The caveat of density-functional theory is that the exact form of $E_{\rm xc}(\rho)$ is unknown.
Consequently, the development of novel approximate exchange-correlation functionals is an important and fundamental task.
One possible approach in the development of new and more generally applicable exchange-correlation functionals is by means of sound mathematical bounds (see e.g. Ref.~\citenum{perdew1996comparison}). 
A particularly useful bound is provided by the {\it Lieb--Oxford inequality}~\cite{Lieb1979,LiebOxford}; it gives a lower bound of the {\it indirect interaction energy} $I_{\rm xc}(\psi) = \langle \psi|\hat V|\psi\rangle-J(\rho_\psi)$. 
In quantum chemistry, it is extensively used as a constraint in the construction and testing of exchange-correlation functionals~\cite{perdew1996generalized,PhysRevLett.78.1396,Perdew,Levy1993,Odashima2009,haunschild2012hyper}. 
Hence, having a tight estimate for the bound is highly desirable.  
For more quantum-chemistry related works see, e.g., Refs.~\citenum{odashima2007tight} and \citenum{vela2009variable}.

The Lieb--Oxford inequality, first formulated in three dimensions~\cite{Lieb1979}, states that the indirect interaction energy for any $N$-particle wave function $\psi$ is bounded from below, viz.,
\begin{equation}
\label{eq:LiebOxford3d}
 I_\mathrm{xc}(\psi) \geq -C_3 \int_{\mathbb R^3} \rho_\psi(x)^{4/3} \mathrm{d}x,
\end{equation}
where $C_3\leq 1.68$ was established by Lieb and Oxford in Ref.~\citenum{LiebOxford}.
The constant $C_3$ was further improved by Chan and Handy~\cite{Garnet} to $C_3\leq 1.64$ and, more recently, $C_3 \geq 1.44$ was derived by Cotar and Petrache\cite{cotar2017equality}, and Lewin, Lieb and Seiringer\cite{lewin2019}.
In the two-dimensional case, a Lieb--Oxford bound has been proven by Lieb, Solovej and Yngvason\cite{Lieb1995} stating that
\begin{equation}
\label{LiebOxford2d}
I_\mathrm{xc}(\psi) \geq -C_2 \int_{\mathbb R^2} \rho_\psi(x)^{3/2} \mathrm{d}x,
\end{equation}
where $C_2 \leq 481$.
In the work of R\"as\"anen et al., an argument based on universal scaling properties was used to conjecture that in the $d$-dimensional case 
\begin{align}
\label{eq:post}
I_\mathrm{xc}(\psi) \geq -C_d \int_{\mathbb R^d} \rho_\psi(x)^{1+\frac 1 d} \mathrm{d}x.
\end{align}
Note that Eq.~\eqref{eq:post} agrees with the proven results for three and two dimensions  (Eqs.~\eqref{eq:LiebOxford3d} and~\eqref{LiebOxford2d}).
Furthermore, using the $d$-dimensional infinite homogeneous electron gas in the low-density limit, Ref.~\citenum{Rasanen} provided further improved bounds for two and three dimensions, and the conjectured one-dimensional bound
$I_\mathrm{xc}(\psi) \geq -C_1 \int_{\mathbb R} \rho_\psi^{2} \mathrm{d}x$.
In order to further improve the Lieb--Oxford bound, Benguria, Bley and Loss introduced an additional term to the right-hand side of Eq.~\eqref{eq:post} that involves the gradient of the single-particle density~\cite{benguria2012new}.
This bound was used to improve the result from Lieb, Solovej and Yngvason in two dimensions~\cite{benguria2012indirect,benguria2012new1}.
Different Lieb--Oxford bounds including density-gradient type corrections were further investigated in Ref.~\citenum{lewin2015improved}.

A crucial observation for the one-dimensional case is that the Coulomb potential $v(r) = r^{-1}$ is   
too singular (using the approach taken here, see Remark~\ref{rmk:Csing}).
Hence, a suitable interaction potential has to be chosen before defining the indirect interaction energy. 
Common examples are the {\it contact potential} (also called {\it Dirac potential}, $v=\eta \delta$), the {\it soft-Coulomb potential} and the {\it regularized Coulomb potential}~\cite{Rasanen,Rasanen2011}, and in the mathematical literature the {\it homogeneous potential} $v(r)= r^{\eps-1}$.
Without mathematical proof, the bounds $I_\mathrm{xc}(\psi) \geq -\frac{\eta}{2} \int_{\mathbb R} \rho_\psi^{2} \mathrm{d}x$ and 
$I_\mathrm{xc}(\psi) \geq -C \int_{\mathbb R} \rho_\psi^{2}\left( K_1 +  \ln\left[ K_2 /(\eps \rho_\psi)\right] \right)  \mathrm{d}x$ 
were reported in Ref.~\citenum{Rasanen} for the contact and soft Coulomb potential, respectively.
This was also confirmed by R\"as\"anen, Seidl and Gori--Giori for finite homogeneous electron gas in the strong interaction limit in Ref.~\citenum{Rasanen2011}.
In the same work, R\"as\"anen et al.~studied the regularized Coulomb potential but did not present an explicit expression for a Lieb--Oxford bound in this case.

We here present a mathematical analysis that addresses several aspects of Lieb--Oxford bounds 
for one-dimensional quantum systems.
This article is structured as follows. 
In Section~\ref{sec:II} we start with the general result by Hainzl and Seiringer~\cite{Hainzl},  which is based on a generalization of the Fefferman--de la Llave decomposition and uses the Hardy--Littlewood maximal function. 
We derive, in Lemma~\ref{th:2}, an alternative to this general bound that does not require the maximal function. This lemma is used in Theorrem~\ref{lem:Ulemma} with further restrictions on the considered potentials (Assumptiom~\ref{ass:1}). 
In Section~\ref{sec:III}, we present Lieb--Oxford bounds for 
approximate Coulomb-type potentials. 
In particular, we consider a convex version of the soft Coulomb potential and derive in Theorem~\ref{thm:csc} Lieb--Oxford bounds with logarithmic terms 
of the particle density. 
This type of terms appear in one-dimensional conductors, also called ultra-thin wires, when modelling interactions with a soft Coulomb potential~\cite{fogler2005}. 
We show that these terms are also included in a Lieb--Oxford bounds for the regularized Coulomb potential and, to the best of our knowledge, present the first explicit expression of  Lieb--Oxford inequalities for this potential.
We thus complement the analyses of R\"as\"anen et al.~by deriving  explicit expressions for different Lieb--Oxford bounds considered in Refs.~\citenum{Rasanen} and~\citenum{Rasanen2011}. 
We also address the conjectured form in Eq.~\eqref{eq:post} with $d=1$ for all here studied potentials. 
In addition, we address a Lieb--Oxford bound for the homogeneous one-dimensional Hubbard model, which finds application in the description of the Luttinger liquid and the Mott insulator~\cite{capelle2003density,lima2003density,schonhammer1995density}.
The Lieb--Oxford bound for this particular model system is derived in Appendix~A. 
The development of applications in low-dimensional physics using density-functional theory 
shows the potential importance of one-dimensional (density-functional) constraints just as in three dimensions~\cite{Gedanken}.
Moreover, one-dimensional Lieb--Oxford bounds are applicable to confined higher-dimensional systems~\cite{Hainzl} and, as noted in Ref.~\cite{Rasanen}, there is a crossover between one- and two-dimensional bounds

A.L. is thankful for useful discussions with S. Di Marino during the BIRS workshop {\it Optimal transport in DFT}. 
%organized by P. Gori-Giorgi, M. Lewin, and B. Pass. 
The support of the Norwegian Research Council through the Grant Nos. 287906 and 262695 (CoE Hylleraas Centre for Quantum Molecular Sciences), and from ERC-STG-2014 Grant Agreement No. 639508 are acknowledged. The authors are very thankful to E.~I. Tellgren for comments and suggestion that helped improve the manuscript and thanks M.~A. Csirik for useful comments. 
We furthermore thank the anonymous reviewers that, in particular, motivated us to derive a Lieb--Oxford bound within the Hubbard model given in Appendix~A.

%%%%%%%%%%%%%%%%%%%%%%%%%%%%%%%%%%%%%%%%%%%%%%%
%
%
%   Lieb--Oxford bound in one dimension
%
%
%%%%%%%%%%%%%%%%%%%%%%%%%%%%%%%%%%%%%%%%%%%%%%

\section{Lieb--Oxford bounds in one dimension}   
\label{sec:II}

\subsection{Prerequisite}
Throughout this article we assume that the $N$-particle wave function $\psi\in L^2((\mathbb{R}\times \{\uparrow,\downarrow \})^N)$ is normalized, i.e., that $\Vert\psi \Vert_{2}=1$ holds. 
The one-body particle density associated with a wave function $\psi$ is defined through 
\begin{equation*}
\begin{aligned}
\rho_\psi(x) = N \sum_{q_i\in \{\uparrow,\downarrow \}} \int_{\mathbb R^{N-1}} \vert \psi(x,q_1,x_2,q_2,...,x_N,q_N)  \vert^2 \\
\times \mathrm{d}x_2...\mathrm{d}x_N.
\end{aligned}
\end{equation*}
Subsequently, we furthermore assume that $\psi$ has finite kinetic energy, i.e., 
\begin{align*}
    \mathcal K(\psi) = \frac{1}{2} \sum_{q_i\in \{\uparrow,\downarrow \}}\int_{\mathbb R^N} \vert \nabla  \psi(x_1,q_1,x_2,q_2,...,x_N,q_N)\vert^2  \\
    \times \mathrm{d}x_1... \mathrm{d}x_N< +\infty. 
\end{align*}
The space of wave functions that fulfill these constraints is the Sobolev space with $L^2$-topology denoted $H^1$.
By the Hoffmann--Ostenhof inequality~\cite{Lieb1983}, $\rho_\psi^{1/2} \in H^1(\mathbb R)$ for $\psi\in H^1$.
Furthermore, if $\rho_\psi^{1/2} \in H^1(\mathbb R)$, the Sobolev inequality in one dimension (see e.g. Theorem~8.5 in Ref.~\citenum{lieb2001analysis})
implies that $\rho_\psi\in L^p(\mathbb R)$ for all $p\in [1,\infty]$.

Let the electron--electron repulsion be modelled by a potential $v:\mathbb{R}^+\rightarrow \mathbb{R}$.
The indirect interaction energy is then defined through ($N\geq 2$)
\begin{equation*}
\begin{aligned}
I_\mathrm{xc}(\psi) = \langle \psi |\sum_{i < j}v(|x_i-x_j|)|\psi\rangle  - D(\rho_\psi,\rho_\psi),  
\end{aligned}
\end{equation*}
where $D(\rho,\rho)$ denotes the {\it direct part} (i.e., the Hartree term $J(\rho) = D(\rho,\rho)$) of the interaction energy, viz.,
\begin{align*}
D(\rho,\rho) = \frac{1}{2}\int_{\mathbb R}\int_{\mathbb R} \rho(x)\rho(y)v(|x-y|)\mathrm{d}x \mathrm{d}y.
\end{align*}

\subsection{General Lieb--Oxford bounds}
We begin by introducing the {\it Hardy--Littlewood maximal function}, 
\begin{equation}\label{eq:Mf}
\begin{aligned}
(\mathcal Mf)(x) = \sup_{r>0 }\Big\{\frac{1}{2r}\int_{|x-y|<r}f(y)\mathrm{d}y \Big\}.
\end{aligned}
\end{equation}
For $1<p<+\infty$ the operator $\mathcal{M}$ is bounded in the $L^p$-topology, $\Vert \mathcal Mf \Vert_p\leq M_p \Vert f \Vert_p$, with 
\begin{equation}
\begin{aligned}
\label{eq:int-pol}
 M_p := \left( \frac{2^p2p}{p-1}\right)^{\frac{1}{p}}.
\end{aligned} 
\end{equation}

Hainzl and Seiringer proved the following general bound 
(for a proof we refer to Lemma~2 in Ref.~\citenum{Hainzl}).

\begin{lemma}
\label{thm:HSlem2}
Let $v(r)$ be convex and $\lim_{r\to \infty}v(r)= 0$. 
Then, for any nonnegative function $\gamma$ on $\mathbb R$, 
\begin{equation}
\label{eq:Hainzl-lem}
\begin{aligned}
I_\mathrm{xc}(\psi) 
\geq 
-\frac 1 2  \int_{\R} \Big(  (\mathcal M\rho_\psi)(x)^2  \int_0^{\gamma(x)} v''(r) r^2\mathrm{d}r  \\
 + (\mathcal M\rho_\psi)(x)\int_{\gamma(x)}^\infty v''(r) r\mathrm{d}r \Big) \mathrm d x. 
\end{aligned}	
\end{equation}
\end{lemma}

\begin{remark}
    If not $\gamma(x)>0$ (almost everywhere) we get in Eq.~\eqref{eq:Hainzl-lem} a contribution $\int_\Omega (\mathcal M\rho_\psi)(x)\mathrm d x $, where 
    $\Omega =\{ x\in\R \mid \gamma(x)=0 \}$. Since there does not exist a general bound of the Hardy--Littlewood maximal function in the $L^1$-topology, unless $\gamma>0$,   Eq.~\eqref{eq:Hainzl-lem} cannot be turned into a bound in terms of~$\rho$.
\end{remark}

Assuming an appropriate $\gamma$ has been chosen in Lemma~\ref{thm:HSlem2}, when transforming the inequality in Eq.~\eqref{eq:Hainzl-lem} to be in terms of the particle density $\rho$ instead of $\mathcal M \rho$, the factor $M_p$ (see  Eq.~\eqref{eq:int-pol}) enters. 
In general, this yields suboptimal bounds. A simple illustration can be given by the choice $\gamma \equiv +\infty$. 
The bound in Eq.~\eqref{eq:Hainzl-lem} then reduces to 
\begin{equation*}
\begin{aligned}
I_\mathrm{xc}(\psi) \geq -8 \int_\mathbb{R} v''(r) r^2\mathrm{d}r \int_{\R}  \rho_\psi(x)^2  \mathrm{d}x .
\end{aligned}
\end{equation*}
If furthermore $\lim_{r\to \infty}v''(r)= \lim_{r\to \infty}v'(r)= 0$, then $\int_\mathbb{R} v''(r) r^2\mathrm{d}r = 2 \int_\mathbb{R} v \mathrm{d} r$ and we obtain
\begin{equation}
\label{eq:HSestimate1D}
I_\mathrm{xc}(\psi) \geq -16 \int_\mathbb{R} v(r) \mathrm{d}r \int_{\R}  \rho_\psi^2  \mathrm{d}x.
\end{equation}
However, by a change of coordinates and using the Cauchy--Schwarz inequality, we find
\begin{equation}
\label{eq:direct}
I_\mathrm{xc}(\psi) \geq - \int_\R v(r) \mathrm{d}r \int_{\mathbb R}  \rho_\psi^2 \mathrm{d}x.
\end{equation}
Compared to Eq.~\eqref{eq:HSestimate1D}, Eq.~\eqref{eq:direct} yields a 16 times tighter constant (since $M_2^2 = 16$).

For constant nonnegative $\gamma$, an alternative result can be established, which, compared to Lemma~\ref{thm:HSlem2}, does not use the maximal function in Eq.~\eqref{eq:Mf}.
(The improvement by Lieb and Oxford~\cite{LiebOxford} of the three-dimensional bound of the indirect interaction energy originally given by Lieb\cite{Lieb1979} included the dispense of the Hardy-Littlewood maximal function.)

\begin{lemma}
\label{th:2}
Let $v$ be convex and $\lim_{r\to \infty}v(r)= 0$. Then, for any constant $\gamma \geq 0$,
\begin{equation} 
\begin{aligned}
\label{eq:lem2}
I_\mathrm{xc}(\psi) 
&\geq 
- \frac 1 2 \int_{\R}\rho_\psi^2 \mathrm{d}x \int_0^\gamma v''(r) r^2\mathrm{d}r\\ 
&\quad - \frac 1 2 \int_\R \rho_\psi \mathrm d x \int_\gamma^\infty v''(r) r\mathrm{d}r.
\end{aligned}
\end{equation}
\end{lemma} 

\begin{proof}
Following the notation in Ref.~\citenum{Hainzl}, we introduce
\begin{equation*}
\begin{aligned}
\alpha_\psi(r,z) 
= \int_{z-r}^{z+r}\rho_\psi(x)\mathrm{d}x
= \int_{\mathbb{R}}b_r(x,z)\rho_\psi(x)\mathrm{d}x,
\end{aligned}
\end{equation*}
where $b_r(x,z)$ is equal to one if $z-r < x < z+r$ and zero elsewhere. 
Using the properties of $v$, from Eq.~(12) in Ref.~\citenum{Hainzl} we obtain
\begin{equation}
\label{eq:2terms}
\begin{aligned}
- I_\mathrm{xc}(\psi) 
&\leq 
\int_{\mathbb R} \int_0^{\frac \gamma 2} v''(2r) \alpha_\psi(r,z)^2 \mathrm{d}r\mathrm{d}z \\
&\quad +  \int_{\mathbb R} \int_{\frac \gamma 2}^\infty v''(2r) \alpha_\psi(r,z) \mathrm{d}r\mathrm{d}z.
\end{aligned}
\end{equation}
The Cauchy--Schwarz inequality then yields
\begin{equation*}
\begin{aligned}
\alpha_\psi(z,r)^2   &= \left( \int_{\mathbb R} \rho_\psi(x)  b_r(x,z)^2 \mathrm{d}x \right)^2 \\
&\leq 2r  \int_{\mathbb R}\rho_\psi(x)^2 b_r(x,z) \mathrm{d} x.
\end{aligned}
\end{equation*}
This implies $\int_\mathbb{R} \alpha_\psi(z,r)^2 \mathrm{d}z \leq (2r)^2 \int_{\mathbb R} \rho_\psi(x)^2 \mathrm{d}x $.
Hence, for the first term on the right-hand side of Eq.~\eqref{eq:2terms} we find
\begin{equation} \label{eq:20m1}
\begin{aligned}
%&
\int_{\mathbb R} \int_0^{\frac \gamma 2} v''(2r) \alpha_\psi(r,z)^2 \mathrm{d}r\mathrm{d}z
\leq \frac 1 2 \int_\R \rho_\psi^2 \mathrm{d}x \int_0^{\gamma} v''(r) r^2 \mathrm{d}r .
\end{aligned}
\end{equation}
For the second term, we obtain
\begin{equation} \label{eq:20m2}
\begin{aligned}
%&
\int_{\mathbb R} \int_{\frac \gamma 2}^\infty  v''(2r) \int_{\mathbb R} b_r(x,z) \rho_\psi(x) \mathrm{d}x \mathrm{d}r\mathrm{d}z\\ 
 \leq \frac 1 2  \int_\R \rho_\psi \mathrm d x \int_{\gamma}^\infty v''(r) r \mathrm{d}r .
\end{aligned}
\end{equation}
Inserting Eqs.~\eqref{eq:20m1} and~\eqref{eq:20m2} into Eq.~\eqref{eq:2terms} gives Eq.~\eqref{eq:lem2} and completes the proof.
\end{proof}

\begin{remark}
Using Lemma~\ref{th:2}, we recover the estimate in Eq.~\eqref{eq:direct} under the same assumptions on $v$, viz., 
Eq.~\eqref{eq:lem2} with $\gamma = +\infty$ reduces to 
\begin{equation*}
I_\mathrm{xc}(\psi) 
\geq -\int_\mathbb{R} v(r)\mathrm{d}r \int_{\mathbb R}\rho_\psi(x)^2 \mathrm{d}x.
\end{equation*}
\end{remark}

Under more restrictions on the potentials, the general bounds above can be further specified. 
To that end, we introduce the following criteria for the first and second moments of $v''$.

\begin{assumption} 
\label{ass:1}
The potential $v: \R_+ \to \R $ satisfies
\begin{itemize}
    \item[(i)] $v$ is convex,
    \item[(ii)] $\lim_{r\to+\infty}v(r) = \lim_{r\to+\infty}v'(r)r= 0$, and 
    \item[(iii)] for $\gamma= \gamma(x) \geq 0$ 
    \begin{equation*}
    \begin{aligned}
    &\int_0^\gamma v''(r) r^2 \mathrm{d}r \leq c_1 \ln(1 + c_2\gamma), \quad c_1,c_2> 0, \\
    &\int_\gamma^\infty v''(r) r \mathrm{d}r \leq c_3 \gamma^{-1}, \quad c_3>0.
    \end{aligned}
    \end{equation*}
\end{itemize}
\end{assumption} 
\begin{remark}\label{rmk:Csing}
    For the Coulomb potential $v(r) = r^{-1}$ we have $\int_0^\gamma v''(r) r^2 \mathrm{d}r=+\infty$.
\end{remark}

\begin{theorem}
\label{lem:Ulemma}
Suppose Assumption~\ref{ass:1}. 
Then 
\begin{align}\label{eq:Ulemma1}
     I_\mathrm{xc}(\psi) &\geq  -8 \int_{\R} \rho_\psi^2  \Bigg[ A_1   + c_1   \ln\left( 1 + \frac{c_2 e^{-3}}{\rho_\psi}  \right) \Bigg]\mathrm{d}x,
\end{align}
where $A_1 =c_1 (\ln(2)+3)   + c_3$.
Moreover, for any $\alpha>0$
\begin{equation}
\label{eq:Ulemma2}
\begin{aligned}
    I_\mathrm{xc}(\psi) &\geq -\frac 1 2 \int_\R\rho_\psi^2 \mathrm{d}x \left[  \frac{Nc_3}\alpha  + c_1   \ln\left(1 + \frac{\alpha c_2}{\int_\R\rho_\psi^2\mathrm{d}x}\right)   \right].
\end{aligned}
\end{equation}

\end{theorem}

\begin{proof}
To prove Eq.~\eqref{eq:Ulemma1}, we use Lemma~\ref{thm:HSlem2} with the choice $\gamma(x) = (\mathcal M\rho_\psi)(x)^{-1}$. 
This yields 
\begin{align*}
    I_\mathrm{xc}(\psi) & \geq 
-\frac 1 2  \int_{\R}  (\mathcal M\rho_\psi)^2 \Bigg[ c_3 +  c_1 \ln \left(1 + \frac{c_2}{\mathcal M\rho_\psi} \right) \Bigg] \mathrm d x \\
&= -\frac 1 2 \int_{\R}\Bigg( (c_1 \ln(2)  + c_3) (\mathcal M\rho_\psi)^2  \\
&\quad + c_1  (\mathcal M\rho_\psi)^2 \left[\ln \frac{c_2}{\mathcal M\rho_\psi}\right]_+ \Bigg) \mathrm{d}x.
\end{align*}
We now proceed similar to Ref.~\onlinecite{Hainzl} and use that  
\begin{equation*}
[\ln f(x)]_+ \leq \inf_{0<s<1/3} \frac{f(x)^s}{e s},
\end{equation*}
where $[\;g(x)\;]_+= \max(g(x),0)$ denotes the positive part of $g(x)$. 
This implies
\begin{equation*}
\begin{aligned}
    -\frac 1 2 & c_1 \int_{\R}(\mathcal M\rho_\psi)^2  \left[\ln \frac{c_2}{\mathcal M\rho_\psi}\right]_+ \mathrm{d}x \\
    & \geq  -\frac 1 2 c_1 \int_{\R}(\mathcal M\rho_\psi)^2  
    \inf_{0<s<1/3} \frac{1}{se}\left( \frac{c_2}{\mathcal M\rho_\psi}\right)^s \mathrm{d}x \\
    &\geq - 8 c_1 \int_{\R} \rho_\psi^2  
    \inf_{0<s<1/3} \frac{1}{se}\left( \frac{c_2}{\rho_\psi}\right)^s \mathrm{d}x \\
    &\geq - 8 c_1 \int_{\R} \rho_\psi^2  \ln\left( e^3 + \frac{c_2}{\rho_\psi}  \right) \mathrm{d}x .
\end{aligned}
\end{equation*}
Hence,
\begin{align*}
    I_\mathrm{xc}(\psi) &\geq-8 \int_{\R}\rho_\psi^2 \Bigg[ (c_1 \ln(2)  + c_3)   
     + c_1 \ln\left( e^3 + \frac{c_2}{\rho_\psi}  \right) \Bigg]\mathrm{d}x .
\end{align*}

To prove Eq.~\eqref{eq:Ulemma2}, we use Lemma~\ref{th:2}. 
Choose $\gamma= \alpha /\int_\R \rho^{2} \mathrm{d}x$. 
Since $v$ fulfills Assumption~\ref{ass:1}, we obtain
\begin{equation*}
\begin{aligned}
I_\mathrm{xc}(\psi) & \geq - \frac 1 2 \int_\R\rho_\psi^2\mathrm{d}x \left[ \frac{N c_3}\alpha  +   c_1 \ln \left(1 + 
\frac{\alpha c_2}{\int_\R\rho_\psi^2 \mathrm{d}x} \right) \right].
\end{aligned}
\end{equation*}
\end{proof}

\begin{remark}
We highlight that one of the terms in the bound in Eq.~\eqref{eq:Ulemma2} depends on the particle number $N$ and furthermore that $\alpha=\alpha(N)$ is a viable choice, see Sec.~\ref{sec:con}, in particular Eq.~\eqref{eq:Ndep1} and \eqref{eq:Ndep2}. 
(Note that an $\alpha>0$ could also be introduced for Eq.~\eqref{eq:Ulemma1} by instead choosing 
$\gamma(x) = \alpha/ \mathcal M \rho_\psi$ in the proof.)
Although an $N$-dependent bound in the three-dimensional case has been postulated~\cite{LiebOxford} and also  approximated~\cite{odashima2009tightened}, the bound presented here in Eq.~\eqref{eq:Ulemma2} differs significantly.
Note that the right hand side in Eq.~\eqref{eq:Ulemma2} can not be bounded for all $N$ and diverges as $N\to \infty$.
This violates property (iv) in Ref.~\citenum{odashima2009tightened} proposed for a valid particle-dependent bound in three dimensions. 
However, we want to emphasize that a closed analytic expression $C_3(N)$ for all $N$ is, to the best of our knowledge, not known at the moment in the three-dimensional case and similar 
investigations for $d=1$ exceed the scope of this manuscript and are left for future work. 
\end{remark}

%%%%%%%%%%%%%%%%%%%%%%%%%%%%%%%%%%%%%%%%%%
%
%
%    Application part
%
% 
%%%%%%%%%%%%%%%%%%%%%%%%%%%%%%%%%%%%%%%%%%

\section{Lieb--Oxford bounds in one dimension for Coulomb-type potentials}
\label{sec:III}

Based on a universal scaling argument, Ref.~\citenum{Rasanen} conjectured that  
the indirect interaction energy in one dimension satisfies
\begin{equation}
\label{eq:gen-bound}
I_\mathrm{xc}(\psi) \geq -C_1 \int_{\mathbb R} \rho_\psi^{2} \mathrm{d}x,
\quad C_1>0. 
\end{equation}
The analysis in Ref.~\citenum{Rasanen} is based on the observation that
\begin{align*}
I_\mathrm{xc}(\psi) \geq -\frac{C_1}{A_1}   E_\mathrm{x}^\mathrm{LDA}(\rho_\psi),
\end{align*}
where $E_\mathrm{x}^\mathrm{LDA}(\rho)= A_1 \int_{\mathbb R} \rho(x)^{2 }\mathrm{d}x$ is the exchange energy for a homogeneous gas in one dimension. 
(Note that $E_\mathrm{xc}(\psi)$ in the notation of Ref.~\citenum{Rasanen} is the exchange-correlation energy that equals $I_\mathrm{xc}(\psi)$ plus the non-negative contribution from correlated kinetic energy, and our $I_\mathrm{xc}(\psi)$ is denoted $W_\mathrm{xc}(\psi)$ in Ref.~\citenum{Rasanen}.)
In the limit $r_s= 1/(2\rho)\to\infty$, the quotient  $I_\mathrm{xc}(\psi)/E_\mathrm{x}^\mathrm{LDA}(\rho_\psi)$ was studied giving an estimate 
for $\lambda_1:= C_1/A_1 $. Knowing $A_1$ would then give an estimate for $C_1$.
This was also confirmed for a finite homogeneous and strictly correlated electron gas in the limit $N\to\infty$ by R\"as\"anen, Seidl and Gori--Giori~\cite{Rasanen2011} (in Ref.~\citenum{Rasanen2011} the notation $\bar\lambda_1$ was used).  
In Refs.~\citenum{Rasanen} and \citenum{Rasanen2011}, bounds for contact, soft Coulomb, and regularized Coulomb potential were studied. 
The soft Coulomb potential, with a softening parameter $\eps>0$, and the regularized Coulomb potential, with parameter $\beta>0$, are given by
\begin{align} 
\label{eq:SCdef}
v_\eps(r) &= \frac{1}{\sqrt{r^2 + \eps^2}}, \\
v_\beta(r) &=\frac{\sqrt{\pi}}{2\beta} e^{r^2/(4\beta^2)}\,\mathrm{erfc}\left(\frac{r}{2\beta}\right), \label{eq:REGdef}
\end{align}
respectively.
Here, ${\rm erfc}(r) = 1- {\rm erf}(r)$ is the {\it complementary error function}. 

Using that $A_1= 1/2$ for the soft Coulomb potential, R\"as\"anen et al.~obtained 
a modified one-dimensional Lieb--Oxford bound~\cite{Rasanen} 
\begin{align}
\label{eq:gen-bound2}
I_\mathrm{xc}(\psi) \geq - C_1 \int_{\R}  \rho_\psi(x)^2 \left(  K_1 + \ln \left(\frac{K_2}{ \eps\rho_\psi(x)}\right) \right) \mathrm{d}x,
\end{align}
with $C_1 =1 $, $K_1= 3 / 2 -\mu$, and $K_2=2/\pi$ where $\mu= 0.577$ is Euler's constant (see also Ref.~\citenum{Rasanen2011} and Eq.~(2) in Ref.~\citenum{fogler2005}).
Based on the physical arguments in both Refs.~\citenum{Rasanen} and~\citenum{Rasanen2011}, the constant $\overline\lambda_1=C_1/A_1=2$ was obtained for contact and soft Coulomb potential.
However, the potentials have different values of the exchange constant $A_1$.

Moreover, Ref.~\citenum{Rasanen2011}
considered the representation of the {\it Yukawa interaction} in an infinite cylindrical wire of radius $\beta$ (see also Ref.~\citenum{giuliani2005quantum}).
This results in the 
regularized Coulomb potential with cutoff parameter $\beta>0$ in Eq.~\eqref{eq:REGdef}. 
Although, the value $\lambda_1=2$ was numerically supported~\cite{Rasanen2011}, no explicit Lieb--Oxford bound was derived for 
this potential.

In this section we  discuss bounds of the conjectured form in Eq.~\eqref{eq:gen-bound}.
Furthermore, for a convex version of the soft Coulomb potential and for the regularized Coulomb potential we prove  
Lieb--Oxford bounds similar to Eq.~\eqref{eq:gen-bound2} and thereby 
complementing the numerical analysis in Refs. \citenum{Rasanen} and \citenum{Rasanen2011}.

\subsection{Convex soft Coulomb and regularized Coulomb potentials}
\label{sec:ConvexAndSoftCoulomb}

We will here make use of Theorem~\ref{lem:Ulemma} that uses Assumption~1.
 To make the soft Coulomb potential $v_\eps$ (with softening parameter $\eps$) convex, we set
$\tilde v_\eps(r) =  v_\eps(r + r_\eps)$, for $\eps>0$, $r\geq 0$, 
where $r_\eps=\eps/\sqrt{2}$ and $v_\eps$ given by Eq.~\eqref{eq:SCdef}.
\begin{figure}%[h!]
    \centering
    \includegraphics[width = 0.48\textwidth]{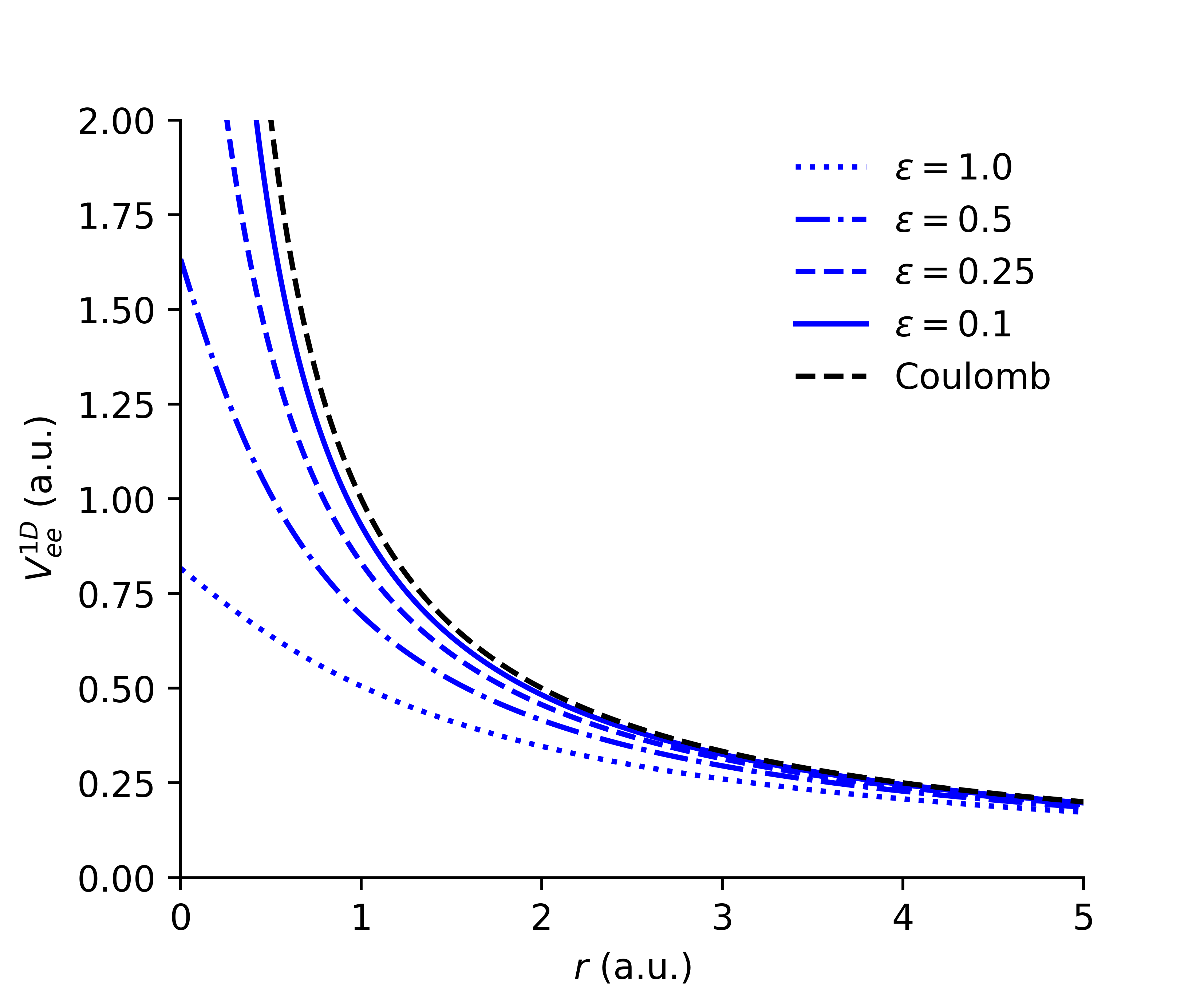}
    \caption{Graphical representation of convex soft Coulomb potential $\tilde v_\eps$ for different softening parameters $\epsilon$ compared to the Coulomb potential $v = 1/r$ itself.}
    \label{fig:CSCPotentials}
\end{figure}
Direct calculation shows that the derivatives satisfy
\begin{equation*}
\begin{aligned}
\tilde v_\eps(r)' &=-\frac{r + r_\eps}{((r+ r_\eps)^2 +\eps^2)^{3/2}}, \\\tilde v_\eps(r)'' &=    \frac{2(r+r_\eps)^2 -\eps^2}{((r+r_\eps)^2 +\eps^2)^{5/2}}.
\end{aligned}
\end{equation*}
Thus, $\tilde v_\eps(r)$ is convex (see Fig.~\ref{fig:CSCPotentials}).

The regularized Coulomb potential is displayed for different values of the cutoff parameter $\beta$ in Fig.~\ref{fig:RegularizedPotentials}. It is straightforward to verify that it is convex for all $\beta>0$. 
Essential for our argument is the following elementary fact about the complementary error function
\begin{equation} \label{eq:erfc-help}
\frac{2}{\sqrt{\pi}} \frac{e^{-r^2}}{r + \sqrt{r^2 + 2}}
\leq 
\mathrm{erfc}(r) 
\leq 
\frac{2}{\sqrt{\pi}} \frac{e^{-r^2}}{r + \sqrt{r^2 + \frac 4 \pi}}.
\end{equation}

\begin{theorem} \label{thm:csc}
    Let $\alpha >0$. The Lieb--Oxford bounds 
\begin{align}\label{eq:IxcA}
          I_\mathrm{xc}(\psi) &\geq  -8 \int_{\R} \rho_\psi^2  \Bigg[ A_1   + c_1   \ln\left( 1 + \frac{c_2 e^{-3}}{\rho_\psi}  \right) \Bigg]\mathrm{d}x, \\
I_\mathrm{xc}(\psi) & \geq - \frac 1 2 \int_\R\rho_\psi^2\mathrm{d}x \left[ \frac{N c_3}\alpha  +   c_1 \ln \left(1 + 
\frac{\alpha c_2}{\int_\R\rho_\psi^2 \mathrm{d}x} \right) \right], \label{eq:IxcB}
\end{align}
hold for the potentials:
    \begin{itemize}
        \item[(i)] Convex soft Coulomb potential $\tilde v_\eps$ with $c_1=c_3=2$, $c_2=c_2(\eps) = \sqrt{2}/\eps$, and $A_1= 2(\ln(2) + 4)$.
        \item[(ii)] Regularized Coulomb potential $v_\beta$ with $c_1=c_3=4$, $c_2=c_2(\beta) = \sqrt{\pi}/(4\beta)$, and $A_1= 4(\ln(2) + 4)$.
    \end{itemize}    
\end{theorem}

\begin{remark}
	Equation \eqref{eq:IxcA} in Theorem~\ref{thm:csc} establishes a bound similar to  Eq.~\eqref{eq:gen-bound2}, with $\ln(\cdot)$ replaced by $\ln(1 + \cdot)$. Note that the inequality $\ln(1 + f) \leq \ln(2) + [\ln(f)]_+$ can be used such that Eq.~\eqref{eq:IxcA} implies
\begin{equation*}
\begin{aligned} 
I_\mathrm{xc}(\psi) 
\geq 
-8 \int_{\R}\Bigg( A_2 \rho_\psi^2 
+ c_1 \rho_\psi^2  \left[\ln\left( \frac{c_2 e^{-3}}{\rho_\psi}  \right)\right]_+ \Bigg)\mathrm{d}x,
\end{aligned} 
\end{equation*}
where $A_2 = A_1 +c_1\ln(2)$.
\end{remark}

\begin{figure}%[h!]
    \centering
    \includegraphics[width = 0.48\textwidth]{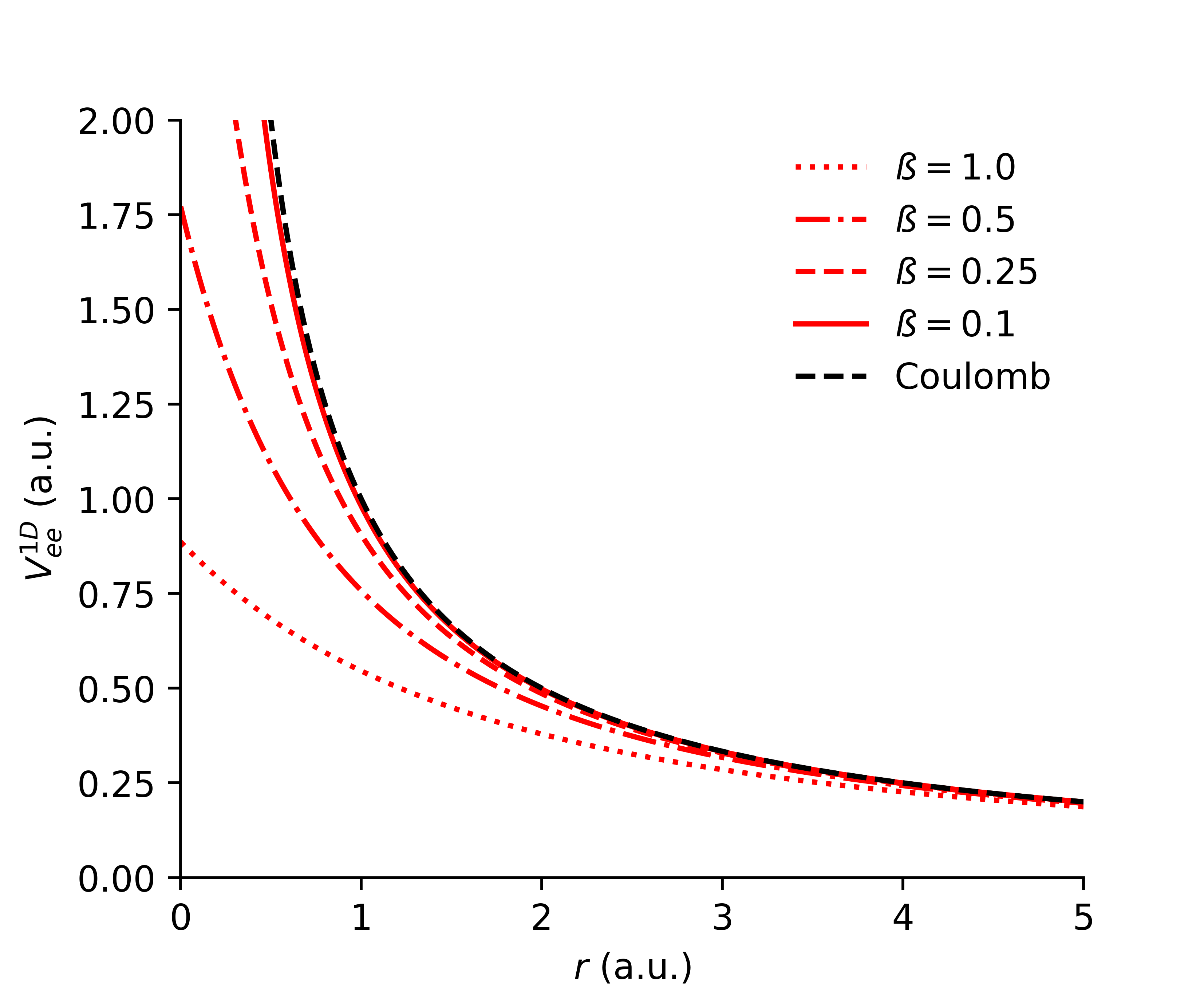}
    \caption{Graphical representation of the regularized Coulomb potential $v_\beta$ for different cutoff parameters $\beta$ compared to the Coulomb potential $v=1/r$ itself.}
    \label{fig:RegularizedPotentials}
\end{figure}

\begin{proof} For both (i) and (ii) we establish the conditions on 
the first and second moment of $v''$ as given in Assumption~\ref{ass:1}.

(i) Let $v$ be the convex soft Coulomb potential $\tilde v_\eps$. 
For the first moment we find
\begin{equation*}
\begin{aligned}
\int_{\gamma}^\infty v'' r \mathrm d r \leq 
\frac{2}{((\gamma +r_\eps)^{2} + \eps^2)^{1/2}}\leq \frac 2 {\gamma}. \label{eq:cscI2}
\end{aligned}
\end{equation*}
Furthermore note that $\int_{0}^{\gamma} v'' r^2 \mathrm d r \leq 2\int_{0}^{\gamma} v \mathrm d r$ such that 
\begin{equation*}
\begin{aligned}
\int_{0}^{\gamma} v'' r^2 \mathrm d r &\leq      2 \left[\ln(\sqrt{r^2 +\eps^2} + r) \right]_{r_\eps}^{\gamma + r_\eps}     \\
& =2 \ln \left( \frac {\gamma+r_\eps}{r_\eps} \frac{ \sqrt{1 +\eps^2/(\gamma+r_\eps)^2} + 1}{\sqrt{3}+1 } \right)   \\
&  \leq 2 \ln \left(  \frac {\gamma+r_\eps}{r_\eps} \right)  
= 2 \ln \left( 1 + \frac {\gamma}{r_\eps} \right) .
\end{aligned}
\end{equation*}
Thus Assumption~\ref{ass:1} holds with the constants $c_1=c_3 = 2$ and $c_2=c_2(\eps) = 2/\eps$. 

(ii) We now consider the case when $v$ is the regularized Coulomb potential $v_\beta$. 
Integrating the second moment of $v''$ by parts, and using that $v'(\gamma) \gamma^2 - 2\gamma v(\gamma)\leq 0$ yields
\begin{equation*}
\int_0^\gamma v''(r)r^2 \mathrm{d}r  \leq 2 \int_0^\gamma v(r) \mathrm{d}r.
\end{equation*}
Employing the upper bound in Eq.~\eqref{eq:erfc-help} we find
\begin{equation} \label{eq:erfc-helphelp}
\begin{aligned}
        v(r) \leq \frac{2}{r + \sqrt{r^2 + \frac{(4\beta^2)}{\pi}}  } 
\end{aligned}
\end{equation}
and therewith
\begin{equation*}
\begin{aligned}
\int_0^\gamma v''(r)r^2 \mathrm{d}r \leq 4 \ln\left(1 + \frac{\sqrt{\pi}\gamma}{4\beta}\right) .
\end{aligned}
\end{equation*}

For the other inequality, 
integrating the first moment of $v''$ by part yields
\begin{equation*}
\int_{\gamma}^\infty v''(r)r \mathrm d r =-\gamma v'(\gamma) + v(\gamma).
\end{equation*}
Using the lower bound in Eq.~\eqref{eq:erfc-help} we obtain
\begin{equation*}
\begin{aligned}
-\gamma v'(\gamma) 
&\leq
\frac{\gamma}{2\beta^2}\left(1-    \frac{2}{1 + \sqrt{1 + \frac{8\beta^2}{\gamma^2}}}\right).
\end{aligned}
\end{equation*}
By the fact that $\sqrt{1+r}\leq 1+r/2$ for $r\geq 0$, we find
\begin{equation*}
\begin{aligned} 
1-\frac{2}{1 + \sqrt{1 + \frac{8\beta^2}{\gamma^2}}} 
\leq 
\sqrt{1 + \frac{8\beta^2}{\gamma^2}}-1 
\leq   
\frac{4\beta^2}{\gamma^2}.
\end{aligned}
\end{equation*}
Hence, $\int_{\gamma}^\infty v''(r)r \mathrm{d}r \leq 4 /\gamma$ since $v(\gamma)\leq 2/\gamma$ 
by Eq.~\eqref{eq:erfc-helphelp}. 
To conclude, Assumption~\ref{ass:1} is fulfilled with $c_1=c_3=4$ and $c_2=c_2(\beta) = \sqrt{\pi}/(4\beta)$.
\end{proof}

\begin{remark}
    Note that $\mathrm{erfc}(r) \leq e^{-r^2} /(\sqrt{\pi} r)$. This yields a tighter constant, $c_3=3$, for the regularized Coulomb potential since it then follows that $v_\beta(\gamma) \leq 1/\gamma$. 
\end{remark}

%%%%%%%%%%%%%%%%%%%%%%%%%%%%%%%%%%%%%%%
%
%
% 
%
%
%%%%%%%%%%%%%%%%%%%%%%%%%%%%%%%%%%%%%%

\subsection{Conjectured bound in one dimension based on the scaling argument}
\label{sec:con}

The conjectured bound in Eq.~\eqref{eq:gen-bound} based on the universal scaling argument of Ref.~\citenum{Rasanen} can be readily obtained for the contact potential.
By simply inserting $v=\delta$ into the direct interaction energy $D(\rho_\psi,\rho_\psi)$, we obtain
\begin{equation}
\label{eq:ContactDirect}
I_\mathrm{xc}(\psi) \geq -D(\rho_\psi,\rho_\psi) = -\frac 1 2 \int_{\mathbb R} \rho_\psi^2 \mathrm{d}x .
\end{equation}
Note that Eq.~\eqref{eq:ContactDirect} was proposed and motivated but the proof not spelled out in Ref.~\citenum{Rasanen}.

Furthermore, we can apply Lemma~\ref{th:2} to (convex) {\it approximate contact potentials} (see Fig.~\ref{fig:ContactPotentials}). 
Let $\sigma>0$ and define
\begin{equation*}
v_\sigma(r) = 
\left\lbrace 
\begin{aligned}
&\frac{2}{\sigma}-\frac{2r}{\sigma^2}, &&0\leq r \leq \sigma, \\
&0, &&\sigma <r.
\end{aligned}
\right.
\end{equation*}
\begin{figure}%[h!]
    \centering
    \includegraphics[width = 0.48\textwidth]{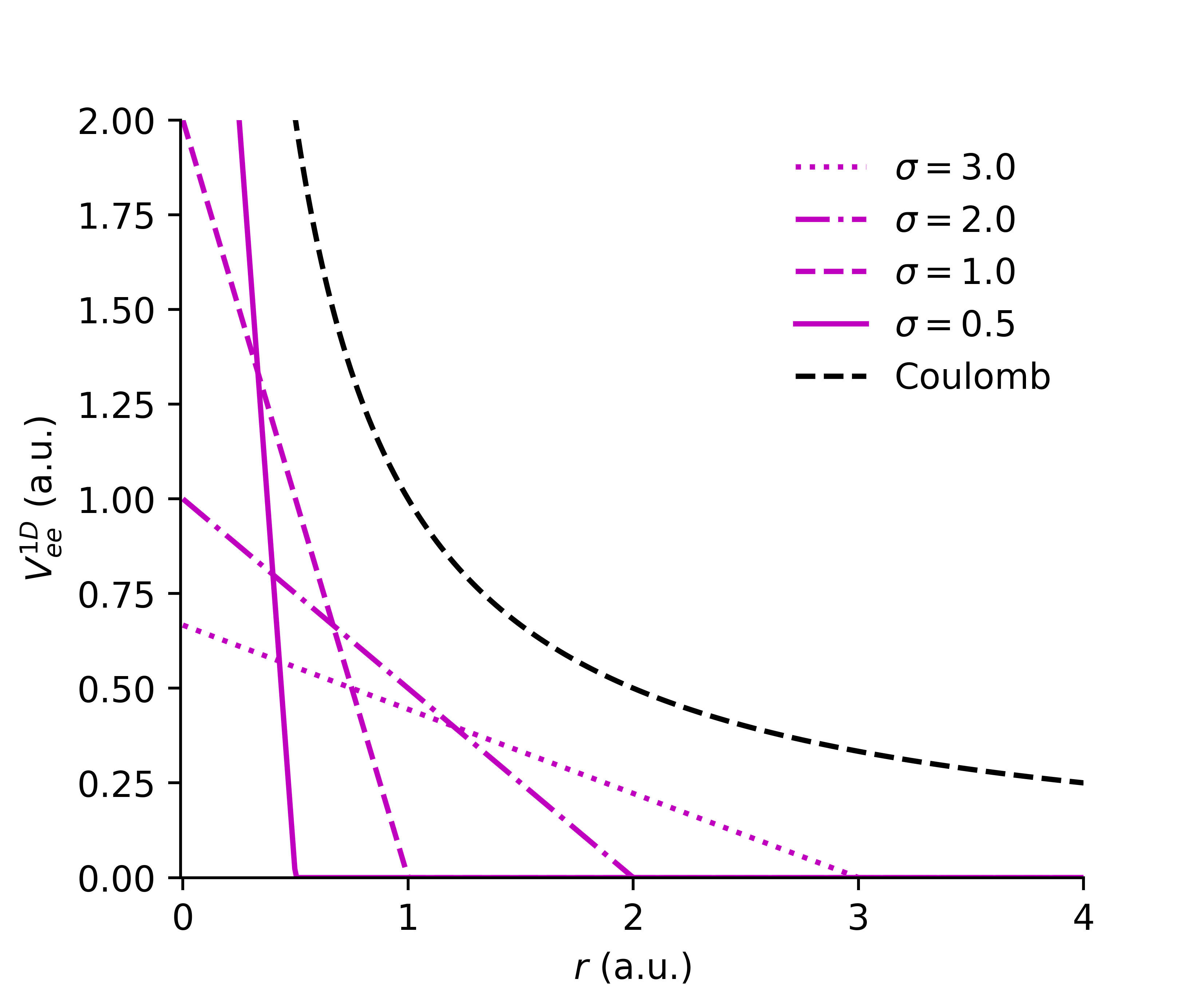}
    \caption{Graphical representation of the approximate contact potential for different approximation parameters $\sigma$ compared to the Coulomb potential itself.}
    \label{fig:ContactPotentials}
\end{figure}
Since $\int_{\mathbb{R}}v_\sigma(r){\rm d}r = 1$, Eq.~\eqref{eq:direct} with $\gamma>\sigma$ yields
\begin{equation*}
I_\mathrm{xc}(\psi)  \geq - \int_{\mathbb R} \rho_\psi^2 \mathrm{d}x .    
\end{equation*}
(Note that $\int_{\mathbb{R}}v_\sigma'' r^2{\rm d}r = 2$ can also be obtained using Eq.~\eqref{eq:lem2}. In that case one either integrates by parts or insert the second order distributional derivative $v_\sigma(r)'' = 2\delta(r-\sigma)/\sigma^2$.)
This bound holds for all $\sigma>0$. Note, however, that in the limit $\sigma\to 0+$, we obtain a nonoptimal constant (twice as large as in Eq.~\eqref{eq:ContactDirect}).

We next discuss the homogeneous potential $v(r) = r^{\eps-1}$. 
In this case Lemma~\ref{thm:HSlem2} with $\gamma=(\mathcal M\rho_\psi)^{-1}$
gives a Lieb--Oxford bound that is arbitrarily close to the conjectured form in Eq.~\eqref{eq:gen-bound}.
(Di Marino has proven similar result in the setting of strictly correlated electrons~\cite{Marino}.) 
This coincides with the general result of Lundholm et al., Lemma~16 in Ref.~\citenum{lundholm2016fractional} with $d=1$, 
viz.,
\begin{align}
I_\mathrm{xc}(\psi) \geq -  \frac{2^{2-\eps}(2-\eps)^2}{\eps(1-\eps)} \int_{\mathbb R} \rho_\psi^{2-\eps} \mathrm{d}x. \label{eq:Hlund}
\end{align}
As the conjectured form $\int_\R\rho_\psi^2\mathrm d x$ is approached, i.e., in the limit $\eps\to0_+$, we have $2^{2-\eps}(2-\eps)^2/(\eps(1-\eps)) \to +\infty$, and where 
the use of the Hardy--Littlewood maximal function introduces the extra factor $2^{3-\eps}(2-\eps)/(1-\eps)$. 
Alternatively, we here note that we can use Lemma~\ref{th:2} to obtain the conjectured form for $v=r^{\eps-1}$ and any $0<\eps<1$, although also with an unbounded constant as $\eps\to0+$.

\begin{theorem}
Let $v(r)=  r^{\eps-1}$ and $0<\eps<1$, then 
\begin{equation*} 
\begin{aligned} 
 I_\mathrm{xc}(\psi) 
&\geq -\left( \frac 1 \eps + \eps - 3\right) \int_{\mathbb R} \rho_\psi^2\mathrm{d}x 
- \left(1-\frac \eps 2 \right)\int_\R \rho_\psi \mathrm{d}x .
\end{aligned}
\end{equation*}
In particular, there exists a $C_\eps>0$ such that 
\begin{align*}
  I_\mathrm{xc}(\psi)  \geq -C_\eps \int_{\mathbb R} \rho_\psi^2\mathrm{d}x.
\end{align*}
\end{theorem}
\begin{proof} 
Note that the homogeneous potential is convex and $\lim_{r\to \infty}v(r)=0$. 
Hence, Lemma~\ref{th:2} is applicable and the choice $\gamma  =1$ gives
\begin{equation*}
\begin{aligned}
\frac{1}{2} \int_{\R}\rho_\psi^2 \mathrm{d}x \int_0^1 v'' r^2\mathrm{d}r 
=
\left( \frac 1 \eps + \eps - 3\right) \int_{\mathbb R} \rho_\psi^2\mathrm{d}x,
\end{aligned}
\end{equation*}
and 
\begin{equation*}
\begin{aligned}
\frac{1}{2}\int_\R \rho_\psi \mathrm d x \int_1^\infty v'' r\mathrm{d}r
&= \left(1-\frac \eps 2 \right)\int_\R \rho_\psi\mathrm{d}x .
\end{aligned}
\end{equation*}
Since $\rho_\psi \in L^1(\mathbb R)$, we find for any $\eps > 0$ an interval $B_\eps$ such that
$\int_\R \rho_\psi \mathrm{d}x \leq \int_{B_\eps} \rho_\psi \mathrm{d}x + N\eps/2$. 
Since $N \geq 2$, we have 
\[
\left(1- \frac \eps 2 \right)\int_\R \rho_\psi\mathrm{d}x > 1
\]
and 
\begin{equation*}
\left(1- \frac\eps 2 \right)\int_\R \rho_\psi\mathrm{d}x 
\leq 
\left(\int_{B_\eps} \rho_\psi \mathrm{d}x \right)^2
\leq 
\vert B_\eps\vert \int_\R \rho_\psi^2 \mathrm{d}x.
\end{equation*}
Defining $C_\epsilon = \vert B_\eps\vert + (\epsilon -3 +1/\epsilon)$ concludes the proof.
\end{proof}
\begin{remark}
As is evident from the proof, the conjectured 
bound for the case $v=r^{\eps-1}$ is not the optimal formulation since for large particle number $\int_\R \rho_\psi \mathrm d x$ is significantly less than $\left(\int_\R \rho_\psi \mathrm{d} x \right)^2$.
\end{remark}

We finalize by commenting on the result in the previous section. 
The Lieb--Oxford bounds established for the convex soft Coulomb potential ($\tilde v_\eps$) and regularized Coulomb potential ($v_\beta$) are not of the conjectured form $I_\mathrm{xc}(\psi) \geq - C_1\int_\R \rho_\psi^2 \mathrm{d}x$.
However, using that $\ln(1 + x) \leq x$ for $x\geq 0$, Theorem~\ref{thm:csc} yields with $\alpha = cN/(c_1c_2)>0$ 
\begin{equation}
\label{eq:LiftedLO}
\begin{aligned}
I_\mathrm{xc}(\psi) 
&\geq 
-  \frac{c_1c_2c_3}{2c} \int_\R\rho_\psi^2 \mathrm{d}x - \frac 1 2  cN .
\end{aligned}
\end{equation}
For the regularized Coulomb potential, where we can take $c_1=c_3 = 4$ and
$c_2(\beta) = \sqrt{\pi}/(4\beta)$, we have 
\begin{align}\label{eq:Ndep1}
     I_\mathrm{xc}(\psi) + \frac 1 2  cN 
&\geq 
-  \frac{2\sqrt{\pi}}{c \beta} \int_\R\rho_\psi^2 \mathrm{d}x,
\end{align}
and similarly for the convex soft Coulomb potential
\begin{align}\label{eq:Ndep2}
      I_\mathrm{xc}(\psi) + \frac 1 2  cN 
&\geq 
-  \frac{2\sqrt{2}}{c \eps} \int_\R\rho_\psi^2 \mathrm{d}x  .
\end{align}
Generally, if the interaction potential is shifted down with a positive constant $v\to v' = v  - c$, 
then the shifted interaction energy, $I_\mathrm{xc}'$, by definition satisfies
$I_\mathrm{xc}' = I_\mathrm{xc}   +  cN/2$. 
In particular, for the regularized Coulomb potential the choice 
$c=v_\beta(0)= \sqrt{\pi}/(2\beta) =:c_\beta $ gives 
$I_\mathrm{xc}'(\psi) \geq -4 \int_\R\rho_\psi^2 \mathrm{d}x$, although the indirect interaction energy $I_\mathrm{xc}'$ does not come from a repulsive potential anymore since $v_\beta':=v_\beta - c_\beta\leq 0$.

\section{Conclusion}

In this article, we have discussed and presented one-dimensional Lieb--Oxford-type inequalities for different Coulomb-like potentials. 
Due to the strong singular nature of the Coulomb potential in one dimension, different choices of pseudopotentials have been employed and investigated. 
Although we were able to derive Lieb--Oxford bounds for a variety of Coulomb-like potentials, we have not found a general or parameter-independent bound. (The contact potential gives the constant $1/2$ but is not reproduced using the other potentials.)
Our strategy follows the general framework of Hainzl and Seringer~\cite{Hainzl},
although an alternative approach (see Lemma~\ref{th:2}) that does not use the Hardy--Littlewood maximal function has been established and subsequently used.
We have focused on potentials that approximate the Coulomb potential and that have been previously addressed in the literature. 
In particular, R\"as\"anen and coauthors~\cite{Rasanen,Rasanen2011} have studied Lieb--Oxford bounds for different interaction potentials.
Reference~\citenum{Rasanen} conjectured that the Lieb--Oxford bound for a one-dimensional system takes the form $I_\mathrm{xc}(\psi)\geq -C_1\int_\R \rho_\psi^2\mathrm{d}x $. 
For the contact potential this can be directly established.
We proved modified results for the soft and the regularized Coulomb potentials (see Section~\ref{sec:ConvexAndSoftCoulomb}) that also involve logarithmic terms of the particle density. %
For the regularized Coulomb potential, the proven Lieb--Oxford inequality in Theorem~\ref{thm:csc} is, to the best of our knowledge, the first known explicit bound for this type of interaction.
In Section~\ref{sec:con}, we additionally investigated to what degree the conjectured bound could be established 
for approximate contact potentials, homogeneous potentials, soft Coulomb and regularized Coulomb potentials. This typically involves unbounded (Lieb--Oxford) constants.
In particular, by applying Lemma~\ref{th:2} we derived a Lieb--Oxford bound for the homogeneous potential $v= r^{\eps-1}$ 
of the conjectured form but with unbounded constant (as $\eps\to 0+$). Using the same lemma, similar results are 
discussed for (convex) soft and regularized Coulomb potentials.

\section*{Data Availability}
Data sharing is not applicable to this article as no new data were created or analyzed in this study.

\appendix

\section{Lieb--Oxford bound for the one-dimensional Hubbard model}

The Hubbard model is used describe the transition between conducting and insulating systems~\cite{hubbard1963electron, altland2010condensed,herring1966magnetism}.
In the one-dimensional case, this can be identified by Hydrogen chains~\cite{essler2005one}. 
It is interesting to note that an expression similar to the conjectured bound in Eq.~\eqref{eq:post} (for the Schr\"odinger model) can be derived for the homogeneous one-dimensional Hubbard model.
The theoretical framework described in Ref.~\citenum{capelle2003density} provides the necessary foundation for the following derivation.

We recall that the Hubbard Hamiltonian in second quantization is
\begin{equation*}
\begin{aligned}
\hat H
&=-t\sum _{\sigma \in \{\uparrow, \downarrow\} }\sum_{i =1}^L
\left(
c_{i,\sigma}^{\dagger } c_{i+1,\sigma }+c_{i+1,\sigma }^{\dagger }c_{i,\sigma }
\right)\\
&\quad+U\sum _{i=1}^L c_{i\uparrow }^\dagger c_{i\uparrow }c_{i\downarrow }^\dagger c_{i\downarrow },
\end{aligned}
\end{equation*}
where $L>0$ is the number of lattice sites, $t>0$ is the hopping integral and $U \geq 0$ is the strength of the on-site interaction that represents the electron repulsion.
In order to describe the exchange-correlation energy, we use an interpolation formula for the ground-state energy of the homogeneous Hubbard model based on the Bethe ansatz~\cite{lieb1994absence} (see e.g. Ref.~\citenum{capelle2003density} for more details).

\begin{figure}
    \centering
    \includegraphics[width = 0.48\textwidth]{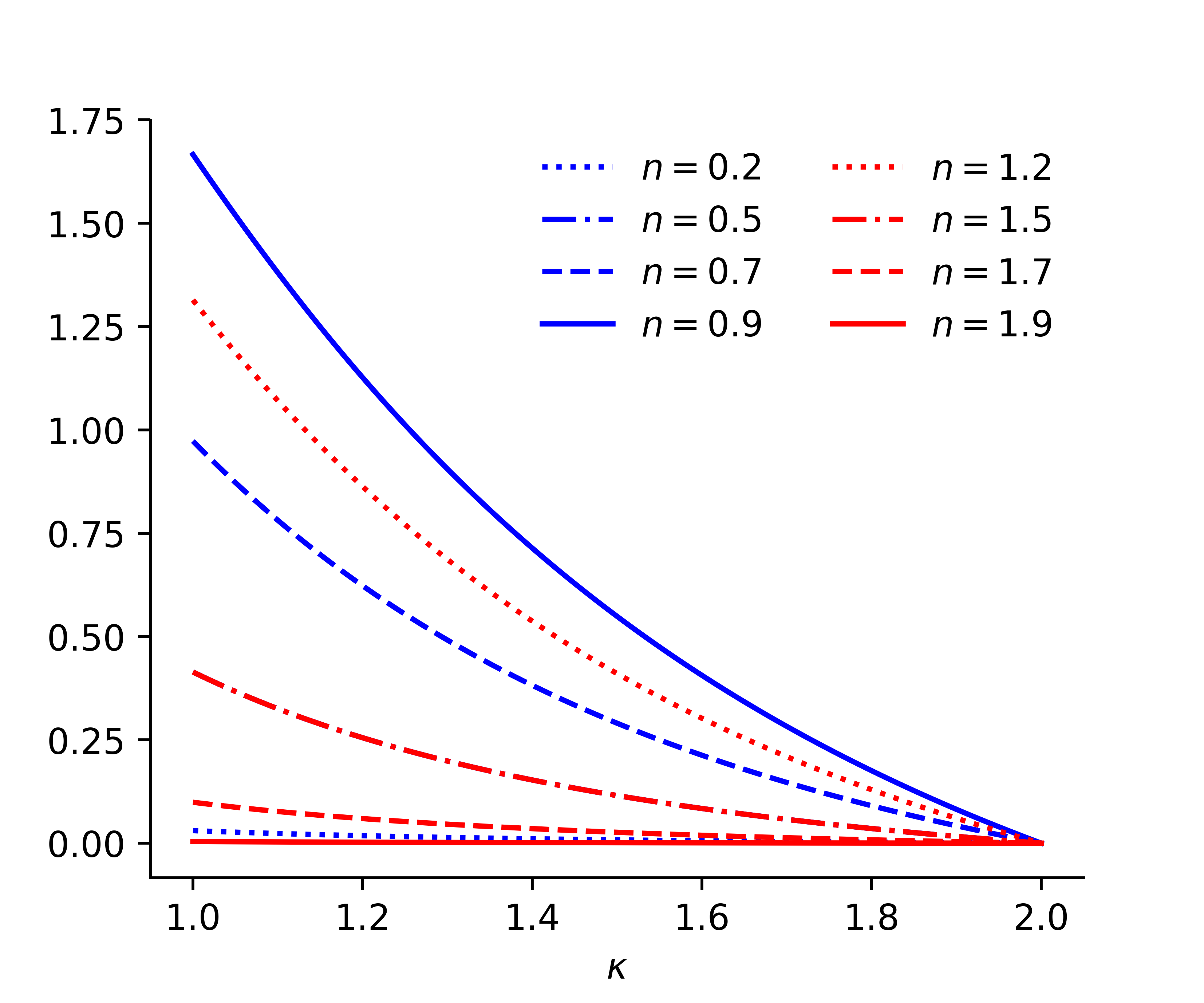}
    \caption{Graphical representation of the the difference $e(n,t,U) - e(n,t,0)$. The blue curves describe the difference for $n\leq 1$ whereas the red curves the difference for $n>1$.}
    \label{fig:F_kappa}
\end{figure}

Considering two different cases for the band filling, i.e., a less than half-filled electronic band ($n = N/L \leq 1$) and a more than half-filled band ($1 < n = N/L \leq 2$), the energy expression reads~\cite{capelle2003density}
\begin{equation*}
e(n,t,U) = 
\left\lbrace 
\begin{aligned}
  - & \frac{2t\beta(U/t)}{\pi} \, \sin 
\left( 
\frac{\pi}{\beta(U/t)}n
\right)  , &&   0 \leq n \leq 1, \\
&  e(2-n,t,U) + U(n-1), && 1 < n \leq 2,
\end{aligned}
\right.
\end{equation*}
where $n = N/L$ describes the band filling and $\beta\in [1,2]$ is a function of the ratio $U/t$.
The exchange-correlation energy is then given by
\begin{equation}\label{eq:A2}
e_{\rm xc}(n,t,U)
=
e(n,t,U) - e(n,t,0) -e_{\rm H}(n,U),
\end{equation}
where $e_{\rm H}$ is the Hartree energy, i.e., $e_{\rm H}(n,U) = U n^2/4$.
Note that $e(n,t,0)$ describes the noninteracting kinetic energy.

\begin{proposition}
Let $t>0$ and $U\geq 0$ be fixed but arbitrary. Then $E_\mathrm{ex}^\mathrm{1DHM}(n,t,U) \geq -\frac{U}{4} \sum_i n_i^2$. 
\end{proposition}
\begin{proof}
We will first demonstrate that $e(n,t,U) - e(n,t,0)\geq 0$.  
Let $0\leq n\leq1$, then 
\begin{align*}
    &e(n,t,U) - e(n,t,0)\\
    & \quad =\frac{2t }{\pi} \left[ 2 \sin\left(  \frac{\pi n}{2} \right)  
    - \beta(U/t) \sin\left( \frac{\pi n}{\beta(U/t)}  \right)\right]  \\
    & \quad =:\frac{2t }{\pi} f_n(\kappa),  
\end{align*}
where $\kappa = \beta(U/t)$. 
We claim that $f_n(\kappa)$ is a positive function on the interval $[1,2]$, see Fig.~\ref{fig:F_kappa}. To show this, we 
just have to note that $f_n'(\kappa) \leq 0$ for $\kappa\in[1,2]$ (using e.g. that $x\leq \tan(x)$ on $[0,\pi/2]$ and that the reverse inequality holds on $[\pi/2,\pi]$) and that (by sine angle addition identity) 
\begin{align*}
    f_n(1) &=  \sin\left(  \frac{\pi n}{2} \right)    
     -  \sin\left( \pi n  \right)  \\
     & =  \sin\left(  \frac{\pi n}{2} \right) - 2 \cos\left(  \frac{\pi n}{2} \right)\sin\left(  \frac{\pi n}{2} \right)  
     \geq 0 = f_n(2).
\end{align*}
For $n>1$ (but less or equal to two), we repeat the above but with $m=2-n$ instead of $n$.
To complete the proof, we use Eq.~\eqref{eq:A2} and that $f_{n}\geq 0$ such that
\begin{align*}
    E_{\rm xc}^\mathrm{1DHM}(n,t,U) &= \sum_{i=1}^L
\left( \frac{2t}{\pi}f_{n_i}(\kappa) -e_{\rm H}(n_i,U) \right) \\
&\geq -\frac{U}{4} \sum_i n_i^2.
\end{align*}
\end{proof}
We note that in case of the homogeneous Hubbard model the lower bound of the indirect interaction energy in terms of the single-particle density takes a discrete form characterized by the site-occupation number $n_i$, i.e., the expectation value of $c_{i\uparrow }^\dagger c_{i\uparrow } + c_{i\downarrow }^\dagger c_{i\downarrow}$.

\bibliography{refs}

\end{document}